%% file: main.tex
\documentclass[
    superscriptaddress,
    floatfix,
    citeautoscript,
]{revtex4-2}

\input{packages}

\begin{document}

\title{Elementary Proof of QAOA Convergence}

\date{\today}

\author{Lennart \surname{Binkowski}}
\email{lennart.binkowski@itp.uni-hannover.de}
\affiliation{Institut für Theoretische Physik, Leibniz Universität Hannover}
\author{Gereon \surname{Koßmann}}
\email{kossmann@stud.uni-hannover.de}
\affiliation{Institut für Theoretische Physik, Leibniz Universität Hannover}
\author{Timo \surname{Ziegler}}
\email{timo.ziegler@volkswagen.de}
\affiliation{Institut für Theoretische Physik, Leibniz Universität Hannover}
\affiliation{Volkswagen AG, Berliner Ring 2, 38440 Wolfsburg}
\author{Ren\'{e} \surname{Schwonnek}}
\email{rene.schwonnek@itp.uni-hannover.de}
\affiliation{Institut für Theoretische Physik, Leibniz Universität Hannover}

\begin{abstract}
The Quantum Alternating Operator Ansatz (QAOA) and its predecessor, the Quantum Approximate Optimization Algorithm, are one of the most widely used quantum algorithms for solving combinatorial optimization problems.
However, as there is yet no rigorous proof of convergence for the QAOA, we provide one in this paper.
The proof involves retracing the connection between the Quantum Adiabatic Algorithm and the QAOA, and naturally suggests a refined definition of the `phase separator' and `mixer' keywords.
\end{abstract}

\maketitle

\section{Introduction}\label{section:Introduction}
\input{Tex/Section1}

\section{Preliminaries}\label{section:Preliminaries}
\input{Tex/Section2}

\section{Convergence Proof for the QAA}\label{section:ConvergenceProofQAA}
\input{Tex/Section3}

\section{Convergence Proof for the QAOA}\label{section:ConvergenceProofQAOA}
\input{Tex/Section4}

\section{Conclusion and Outlook}\label{section:ConclusionAndOutlook}
\input{Tex/Section5}

\begin{acknowledgments}

We thank Tim Heine, Lauritz van Luijk, Tobias J.\ Osborne, Christoph Pohl, Antonio Rotundo, Martin Steinbach, and Reinhard F.\ Werner for helpful discussions.
GK acknowledges financial support by the DAAD and IIT Indore (Kapil Ahuja) for a guest stay. 
RS acknowledges financial support by the Quantum Valley Lower Saxony and by the BMBF project ATIQ.
LB and TZ acknowledge finanical support by the BMBF project QuBRA.

\end{acknowledgments}

\bibliography{main.bib}

\end{document}

%% file: packages.tex
\usepackage[dvipsnames]{xcolor}
\usepackage{amsthm}
\usepackage{thmtools}
\usepackage{mathtools}
\usepackage{physics}
\usepackage{bm} 
\usepackage{bbm} 
\usepackage{braket}
\usepackage{dsfont}
\usepackage{amsmath}
\usepackage{amssymb}
\usepackage{support-caption}
\usepackage{subcaption}
\usepackage{tikz}
\usetikzlibrary{arrows, positioning}
\usepackage{tikz-cd}

        \DeclareMathOperator{\lin}{lin} 
        \DeclareMathOperator{\opt}{opt} 
        \DeclareMathOperator{\spn}{span} 
        \DeclareMathOperator{\init}{I} 
        \DeclareMathOperator{\mixer}{M} 
        \DeclareMathOperator{\phase}{P} 
        \DeclareMathOperator{\dif}{d} 
        \DeclareMathOperator{\dist}{dist} 
        \DeclareMathOperator{\id}{id} 
        \DeclareMathOperator{\im}{im} 
        
        \newcommand*{\N}{\mathbb{N}}

        \newcommand*{\R}{\mathbb{R}}
        \newcommand*{\C}{\mathbb{C}}
        \newcommand*{\hil}{\mathcal{H}}
        \newcommand*{\defdot}{\,\cdot\,} 
        \newcommand*{\defcolon}{\,:\,} 
        
        \newcommand*{\bits}{Z(N)} 
        \newcommand*{\one}{\mathds{1}} 
        \newcommand*{\projection}{\mathds{P}} 
        \newcommand*{\lo}{\mathcal{L}} 
        \newcommand*{\solspace}{\mathcal{S}} 
        \newcommand*{\optspace}{\solspace_{\opt}} 
        \newcommand*{\us}{\ket{+}} 
        \newcommand*{\ub}{U_{B}} 
        \newcommand*{\uc}{U_{C}} 
        \newcommand*{\um}{U_{\mixer}} 
        \newcommand*{\simum}{U_{\mixer, 0}} 
        \newcommand*{\sequm}{U_{\mixer, \sigma}} 
        \newcommand*{\up}{U_{\phase}} 

\usepackage[colorlinks]{hyperref}
\hypersetup{
    colorlinks  = true,
    citecolor   = PineGreen,
    linkcolor   = MidnightBlue,
    urlcolor    = PineGreen
}

\declaretheorem[style=plain]{theorem}

\declaretheorem[style=plain,sibling=theorem]{corollary}
\declaretheorem[style=plain,sibling=theorem]{lemma}
\declaretheorem[style=plain,sibling=theorem]{proposition}
\declaretheorem[style=definition,sibling=theorem]{definition}


%% file: Tex/Section1.tex
In the current era of gate-based noisy quantum computers, the class of \textit{variational quantum algorithms} (VQAs) is at the center of research.
First and foremost, the \textit{quantum approximate optimization algorithm} \cite{Farhi2014AQuantumApproximateOptimizationAlgorithm} receives enormous scientific as well as industrial attention.
Like many other VQAs, it is developed for the purpose of solving \textit{combinatorial optimization problems} (COPs) (maximize $f : \{0, 1\}^{n} \rightarrow \R$ subject to some constraints) on quantum computers with the aid of classical optimizers.
It is, to some extend, a discretized and gate-based version of the \textit{quantum adiabatic algorithm} (QAA, \cite{Farhi2000QuantumComputationByAdiabaticEvolution}) which is itself a continuous-time algorithm.
The QAA and the closely related \textit{quantum annealing} \cite{Kadowaki1998QuantummAnnealingInTheTransverseIsingModel} rely on slowly evolving a quantum system (resp.\ some external parameters) in order to transition a well-known initial state into some state representing an optimal solution.
Due to their analog structure, they are not executable on gate-based architectures, but on \textit{quantum annealers} (see \cite{Hauke2020PerspectivesOfQuantumAnnealingMethodsAndImplementations} for an overview) which constitute the second large family of quantum computer architectures.

In its original formulation, the quantum approximate optimization algorithm is only suited for unconstrained problems.
A common technique for enlarging its scope to constrained problems is \textit{softcoding} the constraints.
That is, the constraints enter the objective function as additional terms, penalizing infeasible inputs.
However, for several instances, this approach was observed to produce unfavorable output distributions which suffer from poor optimization quality or feasibility violation (see, e.g., \cite{Baker2022WassersteinSolutionQualityAndTheQuantumApproximateOptimizationAlgorithmAProtfolioOptimizationCaseStudy,delaGranrive2019KnapsackProblemVariantsOfQAUABatteryRevenueOptimisation,vanDam2021QuantumOptimizationHeuristicsWithAnApplicationToKnapsackProblems}).
In order to improve the treatment of constrained problems Hadfield et al.\ extended the quantum approximate optimization algorithm to the \textit{quantum alternating operator ansatz} (QAOA, \cite{Hadfield2019FromTheQuantumApproximateOptimizationAlgorithmToAQuantumAlternatingOperatorAnsatz}) which also allows for \textit{hardcoding} the constraints.
That is, the objective function is left unchanged and feasibility preservation is instead enforced strictly.

In a nutshell, a QAOA-circuit consists of parametrized \textit{phase separator} gates $\up$ and \textit{mixer} gates $\um$.
Both types of gates should preserve feasibility such that - in an ideal setting - feasible states are mapped to feasible states again.
Classically and iteratively optimizing the circuit parameters then should yield a good approximation of an optimal solution.
This heuristic argument goes through only if every feasible state can be reached:
The QAOA-circuit is, given the right parameter values, able to (approximately) produce every feasible state.
Typically, the reachability of feasible states only depends on the properties of the mixer $\um$.

A more or less rigorous proof why the quantum approximate optimization algorithm should converge for every (unconstrained) COP with only one optimal solution was already given in \cite{Farhi2014AQuantumApproximateOptimizationAlgorithm}.
This sketch of a proof, in turn, builds on the close connection to the QAA and the underlying principle of adiabatic evolution/quantum annealing (see \cite{Morita2008MathematicalFoundationOfQuantumAnnealing} for mathematical treatment).
However, neither is the proof carried out in great mathematical detail, nor does it attempt to be as general as possible.
Moreover, since the QAOA comprises similar principles as the quantum approximate optimization algorithm, it stands to reason to extend this result, once suitably formalized, to the QAOA;
a task that, surprisingly, has not yet been tackled.
With this paper we address this issue and come up with refined definitions for the phase separator and mixer gates which make the connection to the quantum approximate optimization algorithm more visible.

First, we prove the convergence of the QAA with suitable initial Hamiltonian and initial state in \autoref{section:ConvergenceProofQAA}.
The proof is built on the aforementioned proof sketch in \cite{Farhi2014AQuantumApproximateOptimizationAlgorithm}.
We extract the underlying principles and already obtain a precise definition for a \textit{mixer Hamiltonian}.
However, by invoking a version of the adiabatic theorem without \textit{gap condition}, we obtain a more general result which does not require the considered optimization problem to have only one single optimal solution.

Second, we prove the convergence of the QAOA with suitable initial state in \autoref{section:ConvergenceProofQAOA}.
For this, we generalize all the properties of the original mixer proposed in the quantum approximate optimization algorithm.
We define our versions of \textit{simultaneous} and \textit{sequential mixers} which directly make use of the just generalized properties.
The convergence proof is then built on the convergence of the QAA instance which admits the respective mixer Hamiltonian as initial Hamiltonian.
The underlying idea is again due to Farhi et al., but suitably generalized for constrained problems and sequential mixers.

%% file: Tex/Section2.tex
\subsection{\label{subsection:CombinatorialOptimizationProblems}Combinatorial Optimization Problems}

In the following, we restrict to maximization problems, as minimization tasks may be considered analogously.
This choice of the optimization direction simply allows us to state the convergence proofs more compactly.
A generic COP of \textit{size} $N$ is of the form
\begin{align}\label{equation:COP}
    \max_{\bm{z} \in S} f(\bm{z}),\quad S \subseteq \bits,
\end{align}
where $\bits$ denotes the set of bit strings of length $N$, $f : \bits \rightarrow \R$ is the \textit{objective function}, and $S \subseteq \bits$ is the set of \textit{feasible bit strings} or the \textit{solution set}.
The problem is called \textit{unconstrained} if $S = \bits$.
Moreover, we denote the set of all solutions maximizing $f$ by $S_{\max}$.

\subsection{\label{subsection:ProblemEncodingOnQuantumComputers}Problem Encoding on Quantum Computers}

In order to treat a COP with the help of quantum computers, the problem first has to be translated into a quantum-mechanical language.
The standard encoding procedure identifies each bit string $\bm{z}$ with a computational basis state $\ket{\bm{z}}$ of the $N$-qubit space $\hil \coloneqq \C^{2^{N}}$.
The classical objective function $f$ is further considered as an objective Hamiltonian $C$ via
\begin{align}\label{equation:ObjectiveHamiltonian}
    C \coloneqq \sum_{\bm{z} \in \bits} f(\bm{z}) \ketbra{\bm{z}}.
\end{align}
In this setting, the (optimal) solution bit strings span the (\textit{optimal}) \textit{solution space}
\begin{align}\label{equation:SolutionSpace}
    \solspace_{\max} \coloneqq \spn\{\ket{\bm{z}} \defcolon \bm{z} \in S_{\max}\} \subseteq \solspace \coloneqq \spn\{\ket{\bm{z}} \defcolon \bm{z} \in S\} \subseteq \hil.
\end{align}
The maximization task is now equivalent to finding a computational basis state in $\solspace_{\max}$.
By construction, $\solspace_{\max}$ is the eigenspace of $C\vert_{\solspace}$ corresponding to its largest eigenvalue.
In the following, we will slightly relax the quantum optimization task as we will consider any highest energy state of $C\vert_{\solspace}$ an optimal solution.

\subsection{\label{subsection:QuantumAdiabaticAlgorithm}Quantum Adiabatic Algorithm}

In a nutshell, the continuous-time \textit{quantum adiabatic algorithm} (QAA, \cite{Farhi2000QuantumComputationByAdiabaticEvolution}) tackles the eigenstate search via quasi-adiabatic evolution of an initial state $\ket{\iota}$ with respect to a time-dependent Hamiltonian $H(t)$ which interpolates between an initial Hamiltonian $H_{\init}$ and the objective Hamiltonian $C$.
In case of a maximization task, $\ket{\iota}$ should be a highest energy state of $H_{\init}$.
The interpolating Hamiltonian is typically given by the convex combination \footnote{There exist also more sophisticated convex combinations, where the coefficients in front of $H_{I}$ and $C$ are non-linear functions of $t$ (see \cite{Roland2002QuantumSearchByLocalAdiabaticEvolution}).
Allowing for such problem-specific coefficients can improve the convergence rate of the QAA significantly.}
\begin{align}\label{equation:InterpolatingHamiltonian}
    H(t) = H_{\lin(H_{\init}, C)}(t) \coloneqq (1 - t) H_{\init} + t C,\quad t \in [0, 1].
\end{align}
The evolution speed is controlled via a parameter $T > 0$:
The actual time evolution is with respect to $H(s / T)$, $s \in [0, T]$.
The intuition behind the QAA is that evolving a highest energy state of $H(0)$ sufficiently slowly (i.e\@., $T \gg 1$) yields a highest energy state of $C$ if the energy levels stay separated.
Mathematical rigor is granted by the adiabatic theorems (see \autoref{section:ConvergenceProofQAA}).

\subsection{\label{subsection:QuantumApproximateOptimizationAlgorithm}Quantum Approximate Optimization Algorithm}

The \textit{quantum approximate optimization algorithm} \cite{Farhi2014AQuantumApproximateOptimizationAlgorithm} can, in some sense, be seen as a discrete version of the QAA with fixed initial state
\begin{align}\label{equation:QAOAInitialState}
    \us \coloneqq \frac{1}{\sqrt{2^{N}}} \sum_{\bm{z} \in \bits} \ket{\bm{z}}
\end{align}
and initial Hamiltonian
\begin{align}\label{equation:QAOAMixer}
    B \coloneqq \sum_{n = 1}^{N} \sigma_{x}^{(n)}.
\end{align}
Note that $\us$ is the non-degenerate highest energy state of $B$. $B$ and $C$ are incorporated into parametrized gates:
\begin{align}\label{equation:QAOAGates}
    \ub(\beta) \coloneqq e^{- i \beta B} = \prod_{n = 1}^{N} e^{-i \beta \sigma_{x}^{(n)}}\quad \text{and}\quad \uc(\gamma) \coloneqq e^{-i \gamma C}.
\end{align}
Specifying a \textit{depth} $p \in \N$, the parametrized trial states are constructed via
\begin{align}\label{equation:QAOATrialStates}
    \ket{\vec{\beta}, \vec{\gamma}} \coloneqq V(\vec{\beta}, \vec{\gamma}) \us \coloneqq \left(\prod_{q = 1}^{p} \ub(\beta_{q}) \uc(\gamma_{q})\right) \us
\end{align}
In an iterative process, the parameters are updated by a classical optimization rule in order to maximize the expectation value
\begin{align}\label{equation:QAOAExpectationValue}
    F_{p}(\vec{\beta}, \vec{\gamma}) = \braket{\vec{\beta}, \vec{\gamma} | C | \vec{\beta}, \vec{\gamma}}.
\end{align}
Measuring the final outcome $\ket{\vec{\beta}_{\text{opt}}, \vec{\gamma}_{\text{opt}}}$ in the computational basis then yields a distribution of optimal solution approximations.

\subsection{\label{subsection:QuantumAlternatingOperatorAnsatz}Quantum Alternating Operator Ansatz}

Building on the ideas of the quantum approximate optimization algorithm, the \textit{quantum alternating operator ansatz} (QAOA, \cite{Hadfield2019FromTheQuantumApproximateOptimizationAlgorithmToAQuantumAlternatingOperatorAnsatz}) extends its design to general constrained problems.
Given a COP with objective Hamiltonian $C$ and solution space $\solspace$, the parametrized gate $\ub(\beta)$ is substituted with problem-specific `mixer' gates.
For simplicity, we will focus on the case where the same mixers are used in every iteration.
Thereby, we can collect them again in a single mixer gate $\um(\beta)$.
It is demanded to fulfill two important properties:
\begin{itemize}
    \item Feasibility preservation: For all parameter values $\beta \in \R$ $\um(\beta)(\solspace) \subseteq \solspace$ should hold.
    \item Full mixing of solutions: For all feasible computational basis states $\ket{\bm{z}}, \ket{\bm{z}'} \in \solspace$, there should exist a power $r \in \N$ and a parameter value $\beta \in \R$ so that $\braket{\bm{z} | \um^{r}(\beta) | \bm{z}'} \neq 0$.
\end{itemize}

Furthermore, the parametrized gate $\uc(\gamma)$ could be replaced by a more general `phase separator' gate $\up(\gamma)$ which resembles the classical objective function's behavior.
In order to be more concrete, we will further focus on the case where $\um(\beta)$ and $\up(\gamma)$ are given by (products of) exponentials of Hamiltonians.

The correct definition of $\um(\beta)$ follows naturally from the following convergence considerations and is given in \autoref{section:ConvergenceProofQAOA}.
We define the phase separator already now:

\begin{definition}\label{definition:PhaseSeparator}
    Given a COP with solution space $\solspace$ and optimal solution space $\solspace_{\max}$, a Hamiltonian $H$ is called a \textit{phase separator Hamiltonian} iff it fulfills the following two conditions:
    \begin{itemize}
        \item[(i)] $H$ is diagonal in the computational basis.
        \item[(ii)] The eigenspace of $H\vert_{\solspace}$ corresponding to its largest eigenvalue is $\solspace_{\max}$.
    \end{itemize}
    The corresponding (\textit{parametrized}) \textit{phase separator} is given by
    \begin{align}\label{equation:PhaseSeparator}
        \up(H, \gamma) \coloneqq e^{-i \gamma H}.
    \end{align}
\end{definition}

%% file: Tex/Section3.tex
We first examine the convergence behavior of the QAA.
Although originally stated for unconstrained problems, we can easily extend the idea to a COP with a non-trivial solution space $\solspace$:
The initial Hamiltonian $H_{\init}$ should preserve feasibility, i.e., $H(\solspace) \subseteq \solspace$, and the initial state $\ket{\iota}$ should lie within $\solspace$.
In addition, we substitute the objective Hamiltonian $C$ with a more general phase separator Hamiltonian $H_{\phase}$ which trivially preserves feasibility.
Then, for every $t \in [0, 1]$, the time evolution with respect to $H_{\lin(H_{\init}, H_{\phase})}(t)$ applied to $\ket{\iota}$ will give again a feasible state.
Thus, we effectively restrict ourselves to the subspace $\solspace \subseteq \hil$.

The underlying concept of the QAA is captured by the adiabatic theorem.
For our analysis, we use a more general version than Farhi et al.\ did in \cite{Farhi2000QuantumComputationByAdiabaticEvolution}.
\begin{theorem}[Adiabatic Theorem, \cite{Teufel2001ANoteOnTheAdiabaticTheoremWithoutGapCondition}]\label{theorem:AdiabaticTheorem}
    Let $\{H(t) \defcolon 0 \leq t \leq 1\} \subseteq \lo(\hil)$ be a family of self-adjoint operators such that $H(\defdot) \in C^{2}([0, 1], \lo(\hil))$.
    For $T > 0$, let $\Tilde{U}_{T}$ be the solution of
    \begin{align}\label{equation:QuasiAdiabaticEvolution}
        \frac{\dif}{\dif s} \Tilde{U}_{T}(s) = - i H(s / T) \Tilde{U}_{T}(s),\quad 0 \leq s \leq T;\quad \Tilde{U}_{T}(0) = \one
    \end{align}
    and set $U_{T}(t) \coloneqq \Tilde{U}_{T}(t T)$, $0 \leq t \leq 1$.
    Let $\lambda(t)$ be an eigenvalue of $H(t)$, respectively, with corresponding spectral projection $\projection(t)$.
    Furthermore, let $P \in C^{2}([0, 1], \lo(\hil))$ such that for every $0 \leq t \leq 1$, $P(t)$ is a projection with $H(t) P(t) = \lambda(t) P(t)$.
    In addition, $P (t) = \projection(t)$ should hold for almost all $t \in [0, 1]$.
    Then
    \begin{align}\label{equation:AdiabaticLimit}
        \lim_{T \to \infty} (\one - P(t)) U_{T}(t) P(0) = 0
    \end{align}
    uniformly in $t$ in $[0,1]$.
\end{theorem}

\autoref{theorem:AdiabaticTheorem} essentially states that, in the adiabatic limit, starting within (a subspace of) the eigenspace of $H(0)$ corresponding to the eigenvalue $\lambda(0)$, one stays within the eigenspace of $H(t)$ corresponding to the eigenvalue $\lambda(t)$, $0 \leq t \leq 1$, if one follows the time evolution generated by $H$, and the curve of spectral projections $P$ can be $C^{2}$-continued through all potential level crossings.
In contrast, Farhi et al.\ used a version of the adiabatic theorem that prohibits any level crossing (see \cite{Messiah1976QuantumMechanicsVol2}).

A sketch of a convergence proof for the QAA was given in \cite{Farhi2014AQuantumApproximateOptimizationAlgorithm} as an intermediate step to argue the convergence of the quantum approximate optimization algorithm.
Besides the adiabatic theorem, the proof is mainly based on the \hyperref[theorem:PerronFrobenius]{Perron-Frobenius Theorem}.
First recall the definition of irrecudibility in the context of matrices.
\begin{definition}\label{definition:Irreducibility}
    A matrix $A \in \C^{n \times n}$ is called \textit{irreducible} iff there are no proper $A$-invariant coordinate subspaces of $\C^{n}$.
    That is, the only coordinate subspaces left invariant by $A$ are $\{0\}$ and $\C^{n}$.
\end{definition}

\begin{theorem}[Perron-Frobenius]\label{theorem:PerronFrobenius}
    Let $A \in \C^{n \times n}$ be component-wisely non-negative and irreducible. Then $A$ admits a non-degenerate largest eigenvalue.
\end{theorem}

The crucial observation is that the matrix representation of the initial Hamiltonian \eqref{equation:QAOAMixer} in the computational basis fulfills both requirements of the \hyperref[theorem:PerronFrobenius]{Perron-Frobenius Theorem}.
As this will also play an essential role throughout our convergence proof, we use these very properties for giving a first definition of a \textit{mixer}.

\begin{definition}\label{definition:Mixer}
    A Hamiltonian $B \in \lo(\hil)$ is called a \textit{mixer} for a COP with solution space $\solspace$ iff $B(\solspace) \subseteq \solspace$ and $B\vert_{\solspace} \in \lo(\solspace)$ is component-wise non-negative and irreducible in the computational basis.
\end{definition}

The idea is now to apply the \hyperref[theorem:PerronFrobenius]{Perron-Frobenius Theorem} to the linear interpolation $H_{\lin(B, C)}(t)$ at every time $0 \leq t < 1$ to conclude the existence of an eigenvalue curve $\lambda_{\max}$ that connects both the largest eigenvalues of $H_{\lin(B, C)}\vert_{\solspace}(0) = B\vert_{\solspace}$ and $H_{\lin(B, C)}\vert_{\solspace}(1) = C\vert_{\solspace}$.
For this, we need the following immediate result which can be proven quite easily.
\begin{corollary}\label{corollary:IrreducibleMatrices}
    Let $A \in \C^{n \times n}$ be diagonal and let $B \in \C^{n \times n}$ be irreducible.
    Then also $A + B$ is irreducible.
\end{corollary}

\begin{theorem}[Convergence of QAA]\label{theorem:ConvergenceOfQAA}
    Consider a COP with solution space $\solspace \subseteq \hil$, optimal solution space $\optspace \subseteq \solspace$, and phase separator Hamiltonian $C$.
    If $B \in \lo(\hil)$ is a mixer Hamiltonian in the sense of \autoref{definition:Mixer} and $\ket{\iota} \in \solspace$ is a highest energy state of $B\vert_{\solspace}$, then
    \begin{align}\label{equation:AdiabaticEvolutionLinearInterpolation}
        \lim_{T \to \infty} U_{T}(1) \ket{\iota} \in \optspace,
    \end{align}
    where $U_{T}$ is the quasi-adiabatic evolution w.r.t.\ to the linear interpolation between $B$ and $C$.
\end{theorem}

In the following proof, we directly identify all appearing operators $\lo(\solspace)$ with their matrix representation in the computational basis.

\begin{proof}
    Denote by $\lambda_{\max}(t)$ the largest eigenvalue of $H_{\lin(B, C)}\vert_{\solspace}(t)$, for $0 \leq t \leq 1$, respectively.
    Let $0 \leq t_{0} < 1$.
    Since $B\vert_{\solspace}$ is irreducible, so is $(1 - t_{0}) B\vert_{\solspace}$.
    As $C$ is diagonal, also $H_{\lin(B, C)}(t_{0})\vert_{\solspace} = (1 - t_{0}) B\vert_{\solspace} + t_{0} C\vert_{\solspace}$ is irreducible by \autoref{corollary:IrreducibleMatrices}.
    W.l.o.g.\ assume that $C$ has non-negative spectrum \footnote{
        Otherwise consider $\tilde{B} \coloneqq B - \lambda_{\min} \one$ and $\tilde{C} = C - \lambda_{\min} \one$, where $\lambda_{\min} < 0$ is the smallest eigenvalue of $C$.
        Then $H_{\lin(\tilde{B}, \tilde{C})} = H_{\lin(B, C)} - \lambda_{\min}\one$ generates the same time evolution as $H_{\lin(B, C)}$ up to a global phase.
    }.
    Then, $B\vert_{\solspace}$ as well as $C\vert_{\solspace}$ are component-wisely non-negative.
    In summary, $H_{\lin(B, C)}(t_{0})\vert_{\solspace}$ is component-wisely non-negative and irreducible.
    According to the \hyperref[theorem:PerronFrobenius]{Perron-Frobenius Theorem}, $\lambda_{\max}(t_{0})$ is non-degenerate.\\
    Furthermore, the mapping 
    \begin{align}\label{equation:LinearInterpolatingMapping}
        H_{\lin(B, C)}\vert_{\solspace}: \R \rightarrow \lo(\solspace),\quad t \mapsto H_{\lin(B, C)}\vert_{\solspace}(t) = (1 - t) B\vert_{\solspace} + t C\vert_{\solspace} 
    \end{align}
    is analytic and $H_{\lin(B, C)}\vert_{\solspace}(t)$ is symmetric for all $t \in \R$.
    Let $L$ denote the discrete set of level crossings/eigenvalue splittings of $H_{\lin(B, C)}\vert_{\solspace}$.
    According to \cite[Theorem 6.1]{Kato1995PertubationTheoryForLinearOperators}, the instantaneous eigenvalues of $H_{\lin(B, C)}\vert_{\solspace}(t)$, $t \in \R$, can be sorted as $\{\lambda_{m}(t) \defcolon 1 \leq m \leq M\}$, $M \leq 2^{N}\}$, such that $[t \mapsto \lambda_{m}(t)] \in C^{\,\omega}(\R, \R)$ and for the corresponding spectral projections $\projection_{m}(t)$, it holds that $[t \mapsto \projection_{m}(t)] \in C^{\,\omega}\big(\R \setminus L, \lo(\solspace)\big)$, for every $1 \leq m \leq M$.
    Furthermore, the spectral projections have removable singularities in $L$, i.e.\ there exist analytic continuations $P_{m}$, defined on whole $\R$, such that $P_{m}(t) = \projection_{m}(t)$ for $t \in \R \setminus L$, for all $1 \leq m \leq M$.
    By continuity, these continuations are themselves orthogonal projections with constant rank and fulfill 
    \begin{align*}
        H_{\lin(B, C)}\vert_{\solspace}(t) P_{m}(t) = \lambda_{m}(t) P_{m}(t)
    \end{align*}
    for all $t \in \R$.
    W.l.o.g.\ assume $\lambda_{1}(0) = \lambda_{\max}(0)$.
    Since $\lambda_{\max}(t_{0})$ remains non-degenerate for $0 \leq t_{0} < 1$, it follows that $\lambda_{1} \equiv \lambda_{\max}$ on $[0, 1)$ and by continuity of $\lambda_{1}$ that $\lambda_{1} \equiv \lambda_{\max}$ on $[0, 1]$.
    In addition, the corresponding spectral projection $\projection_{1}$ is well-defined on $[0, 1)$. Therefore, its continuation $P_{1}$ fulfills all properties necessary to apply \autoref{theorem:AdiabaticTheorem}, i.e.\ \eqref{equation:AdiabaticLimit} holds.
    Since $P_{1}(0) = \projection_{1}(0)$, one especially obtains that 
    \begin{align*}
        &\ 0 = \lim_{T \to \infty} (\one - P_{1}(1)) U_{T}(1) P_{1}(0) \ket{\iota} =  \lim_{T \to \infty} (\one - P_{1}(1)) U_{T}(1) \ket{\iota} \\
        \Leftrightarrow\quad & \lim_{T \to \infty} U_{T}(1) \ket{\iota} = P_{1}(1) \lim_{T \to \infty} U_{T}(1) \ket{\iota}.
    \end{align*}
    Since $P_{1}(1)$ is a projection with $H_{\lin(B, C)}\vert_{\solspace}(1) P_{1}(1) = \lambda_{\max}(1) P_{1}(1)$, one concludes \eqref{equation:AdiabaticEvolutionLinearInterpolation}.
\end{proof}

Following the above proof, one realizes that the eigenvalue curve $\lambda_{\max}$ does not cross any other eigenvalue curve of $H_{\lin(B, C)}\vert_{\solspace}$ except, possibly, at $t = 1$.
In \cite{Farhi2014AQuantumApproximateOptimizationAlgorithm}, even a level crossing at $t = 1$ is avoided by assuming that the COP only has one optimal solution, implying that $\lambda_{\max}(1)$ is non-degenerate.
However, by invoking a more general version of the adiabatic theorem, we were able to get rid of this assumption.

\begin{figure}
\begin{tikzpicture}[scale=2]
  \draw[->] (0, 0) -- ++(3, 0) node[right] {$t$};
  \draw (0, 0) -- ++(0,-0.05) node[below] {$0$};
  \draw (2.8, 0) -- ++(0,-0.05) node[below] {$1$};
  \draw[->] (0, 0) -- (0, 3) node[above] {$\{\lambda_{i}(t)\}_{i}$};
  \draw[scale=2.8, domain=0:1, smooth, variable=\x, Red] plot ({\x}, {0.8-0.2*sin(2.5*\x*\x r)});
  \draw[scale=2.8, domain=0:1, smooth, variable=\y] plot ({\y*\y}, {0.3+0.2*\y*\y*\y});
  \node (A2) at (1.5,2.1) {\footnotesize \color{Red}$\lambda_{\max}$};
\end{tikzpicture}
\begin{tikzpicture}[scale=2]
  \draw[->] (0, 0) -- ++(3, 0) node[right] {$t$};
  \draw (0, 0) -- ++(0,-0.05) node[below] {$0$};
  \draw (2.8, 0) -- ++(0,-0.05) node[below] {$1$};
  \draw[->] (0, 0) -- ++(0, 3) node[above] {$\{\lambda_{i}(t)\}_{i}$};
  \draw[scale=2.8, domain=0:1, smooth, variable=\x, Red] plot ({\x}, {0.8-0.2*sin(2.5*\x*\x r)});
  \draw[scale=2.8, domain=0:1, smooth, variable=\y] plot ({\y}, {0.3+0.38*\y*\y*\y});
  \node (A2) at (1.5,2.1) {\footnotesize \color{Red}$\lambda_{\max}$};
\end{tikzpicture}
\caption{The eigenvalue curve $\lambda_{\max}$ stays separated from all the other eigenvalue curves for $0 \leq t < 1$.
If the corresponding COP has exactly one optimal solution the separation extends to $t = 1$ (left plot).
However, if the COP has multiple optimal solutions $\lambda_{\max}$ intersects with at least one other eigenvalue curve at $t = 1$ (right plot).}
\label{figure:LevelCrossing}
\end{figure}
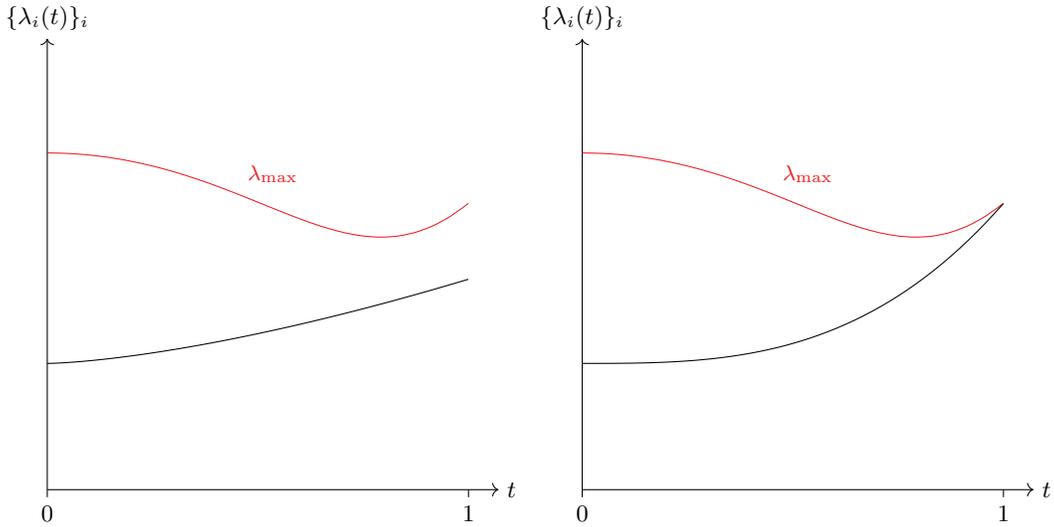

%% file: Tex/Section4.tex
We next examine the convergence behavior of the QAOA which contains the quantum approximate optimization algorithm as a special case.
Its ingredients are basically the same as for our generalized version of the QAA.
However, the decomposition of the mixer Hamiltonian into local Hamiltonians is extremely valuable from an application-oriented point of view and is also introduced by the QAOA.
In the spirit of \autoref{definition:Mixer}, we propose the following adaptation of Hadfield et al.'s\ definition.

\begin{definition}\label{definition:MixingFamily}
    Given a COP with solution space $\solspace$, a family of Hamiltonians $\{B_{i}\}_{i \in I} \subset \lo(\hil)$ is called a \textit{mixing family} iff for every $i \in I$, $B_{i}(\solspace) \subseteq \solspace$, $B_{i}\vert_{\solspace}$ is component-wise non-negative in the computational basis, and any coordinate subspace of $\solspace$ that is left invariant under every $B_{i}$ is already trivial.
\end{definition}

That \autoref{definition:MixingFamily} really is a decomposed version of \autoref{definition:Mixer} can be argued as follows:
Consider the matrix representation of each of the operators $B_{i}\vert_{\solspace}$ in the computational basis as adjacency matrix of a graph whose vertices are identified with feasible computational basis states.
Starting from the graph resembled by $B_{1}$, adding another operator $B_{i}$ corresponds to adding edges represented by non-zero entries of $B_{i}$'s matrix representation.
The actual weights (i.e., values of the entries) are not important, but the condition of component-wise non-negativity implies that no entries are cancelled during the summation, that is, the edge set of the graph $G_{I}$ with adjacency matrix
\begin{align}
    B_{I}\vert_{\solspace}, \quad B_{I} \coloneqq \sum_{i \in I} B_{i}
\end{align}
really is the union of all the edge sets of the graphs $G_{i}$ with respective adjacency matrix $B_{i}\vert_{\solspace}$, $i \in I$.
The imposed condition of triviality of mutual invariant coordinate subspaces then is equivalent to the fact that $G_{I}$ is fully connected which, in turn, is equivalent to its adjacency matrix being irreducible.
Thus, we have concluded
\begin{proposition}\label{proposition:EquivalenceOfMixingProperties}
    Given a COP with solution space $\solspace$, a family of Hamiltonians $\{B_{i}\}_{i \in I} \subset \lo(\hil)$ is a mixing family iff $B_{I}$ is a mixer Hamiltonian.
\end{proposition}

\begin{figure}
\begin{tikzpicture}
    \node (z11) at (0,0) {$\ket{\bm{z}_{1}}$};
    \node (z21) at (2.5,0) {$\ket{\bm{z}_{2}}$};
    \node (z31) at (0,-2.5) {$\ket{\bm{z}_{3}}$};
    \node (z41) at (2.5,-2.5) {$\ket{\bm{z}_{4}}$};
    
    \draw (z11.east) -- (z21.west);
    \draw (z31.east) -- (z41.west);
    
    \node (g1) at (1.25,-3.5) {$G_{1}$};
    
    \node (union) at (4,-1.25) {\Large $\cup$};
    
    \node (z12) at (5.5,0) {$\ket{\bm{z}_{1}}$};
    \node (z22) at (8,0) {$\ket{\bm{z}_{2}}$};
    \node (z32) at (5.5,-2.5) {$\ket{\bm{z}_{3}}$};
    \node (z42) at (8,-2.5) {$\ket{\bm{z}_{4}}$};
    
    \node (g2) at (6.75,-3.5) {$G_{2}$};
    
    \draw (z12.south) -- (z32.north);
    \draw (z32.east) -- (z42.west);
    
    \node (equals) at (9.5,-1.25) {\Large $=$};
    
    \node (z13) at (11,0) {$\ket{\bm{z}_{1}}$};
    \node (z23) at (13.5,0) {$\ket{\bm{z}_{2}}$};
    \node (z33) at (11,-2.5) {$\ket{\bm{z}_{3}}$};
    \node (z43) at (13.5,-2.5) {$\ket{\bm{z}_{4}}$};
    
    \draw (z13.east) -- (z23.west);
    \draw (z13.south) -- (z33.north);
    \draw (z33.east) -- (z43.west);
    
    \node (g3) at (12.25,-3.5) {$G_{I}$};
    
\end{tikzpicture}
\caption{Illustration of the proof of \autoref{proposition:EquivalenceOfMixingProperties}.
The feasible subspace $\solspace$ is spanned by $\ket{\bm{z}_{1}}, \ket{\bm{z}_{2}}, \ket{\bm{z}_{3}}$, and $\ket{\bm{z}_{4}}$. The proper invariant coordinate subspaces of $B_{1}\vert_{\solspace}$ are $\spn\{\ket{\bm{z}_{1}}, \ket{\bm{z}_{2}}\}$ and $\spn\{\ket{\bm{z}_{3}}, \ket{\bm{z}_{4}}\}$ while $B_{2}\vert_{\solspace}$ has $\spn\{\ket{\bm{z}_{1}}, \ket{\bm{z}_{3}}, \ket{\bm{z}_{4}}\}$ and $\spn\{\ket{\bm{z}_{2}}\}$ as its proper invariant coordinate subspaces.
Thus, the set of their mutual proper invariant coordinate subspaces is empty, hence $B_{I} = B_{1} + B_{2}$ is a mixer.
}
\label{figure:ProofSketch}
\end{figure}

Utilizing our definition of a mixing family, we now introduce our version of `simultaneous' and `sequential' mixers.

\begin{definition}\label{definition:Mixer2}
    Let $\mathsf{H} = \{H_{i}\}_{i \in I} \subset \lo(\hil)$ be a mixing family for a given COP.
    The corresponding (\textit{parametrized}) \textit{simultaneous mixer} is defined as
    \begin{align}\label{equation:SimultaneousMixer}
        \simum(\mathsf{H}, \beta) \coloneqq e^{- i \beta \sum_{i \in I} H_{i}}.
    \end{align}
    Specifying a permutation $\sigma \in S(I)$, the corresponding (\textit{parametrized}) \textit{sequential mixer} is defined as
    \begin{align}\label{equation:SequentialMixer}
        \sequm(\mathsf{H}, \beta) \coloneqq \prod_{i \in I} e^{- i \beta H_{\sigma(i)}}.
    \end{align}
\end{definition}

From their definition it immediately follows that both \eqref{equation:SimultaneousMixer} and \eqref{equation:SequentialMixer} fulfill the original QAOA demands: feasibility preservation and full mixing of solutions.
However, due to our refined definition, we can now extend the sketch of a convergence proof in \cite{Farhi2014AQuantumApproximateOptimizationAlgorithm} to the general QAOA setting.
The procedure is as follows:
\begin{itemize}
    \item[1.] discretize the quasi-adiabatic time evolution $U_{T}$
    \item[2.] decompose $H_{\lin(B, C)}$ using a (multivariate) Lie product formula
    \item[3.] exploit the convergence of the corresponding QAA instance
\end{itemize}

We start with a simple statement about the distance of products of operators with factors being close together.

\begin{lemma}\label{lemma:ProductApproximation}
    For $\varepsilon >0$ and $m \in \N$, let $\{V_{j}\}_{j = 1}^{m}, \{W_{j}\}_{j = 1}^{m}, \subset \lo(\hil)$ be families of unitary operators so that
    \begin{align}\label{equation:EpsilonEstimate}
        \left\lVert V_{j} - W_{j}\right\rVert < \varepsilon
    \end{align}
    holds for all $j \in [m]$.
    Then the following estimate is valid:
    \begin{align}\label{equation:ProductApproximation}
        \left\lVert\prod_{j = 1}^{m} V_{j} - \prod_{j = 1}^{m} W_{j}\right\rVert < (1 + \varepsilon)^{m} - 1.
    \end{align}
\end{lemma}

\begin{proof}
    Since \eqref{equation:EpsilonEstimate} holds, one can find linear operators $R_{j} \in \lo(\hil)$ with $\lVert R_{j}\rVert \leq 1$ and $V_{j} = W_{j} + \varepsilon R_{j}$ for each $j \in [m]$, respectively.
    \eqref{equation:ProductApproximation} clearly holds for $m = 1$.
    Therefore, it remains to show that if \eqref{equation:ProductApproximation} holds for an $m \in \N$, then it also holds for $m + 1$: 
    \begin{align*}
        \left\lVert \prod_{j = 1}^{m + 1} V_{j} - \prod_{j = 1}^{m + 1} W_{j}\right\rVert &= \left\lVert \left(\prod_{j = 1}^{m} V_{j}\right) V_{m + 1} - \left(\prod_{j = 1}^{m} W_{j}\right) W_{m + 1}\right\rVert \\
        &= \left\lVert\left(\prod_{j = 1}^{m} V_{j}\right) (W_{m + 1} + \varepsilon R_{m + 1}) - \left(\prod_{j = 1}^{m} W_{j}\right) W_{m + 1}\right\rVert \\
        &= \left\lVert\left(\prod_{j = 1}^{m} V_{j}\right) (\one - \varepsilon R_{m + 1} W_{m + 1}^{*}) - \left(\prod_{j = 1}^{m} W_{j}\right)\right\rVert \\
        &\leq \left\lVert\prod_{j = 1}^{m} V_{j} - \prod_{j = 1}^{m} W_{j}\right\rVert + \left\lVert\left(\prod_{j = 1}^{m} V_{j}\right) \varepsilon R_{m + 1} W_{m + 1}^{*}\right\rVert \\
        &< (1 + \varepsilon)^{m} - 1 + \varepsilon \\
        &\leq (1 + \varepsilon)^{m} - 1 + \varepsilon (1 + \varepsilon)^{m} = (1 + \varepsilon)^{m + 1} - 1.
    \end{align*}
\end{proof}

\begin{theorem}[Convergence of QAOA]\label{theorem:ConvergenceOfQAOA}
    Consider a COP with solution space $\solspace \subseteq \hil$, optimal solution space $\optspace \subseteq \solspace$, phase separator Hamiltonian $C$, and mixing family $\{B_{i}\}_{i \in I}$.
    Let $\up$ and $\um$ be the corresponding phase separator and (simultaneous or sequential) mixer.
    Furthermore, let $\ket{\iota} \in \solspace$ be a highest energy state of $B_{I}\vert_{\solspace}$.
    Then, for every $\varepsilon > 0$, one can choose finitely many parameters $\vec{\beta}$ and $\vec{\gamma}$ such that
    \begin{align}\label{equation:ConvergenceOfQAOA}
        \dist(\ket{\vec{\beta}, \vec{\gamma}}, \optspace) < \varepsilon,
    \end{align}
    where
    \begin{align}
        \ket{\vec{\beta}, \vec{\gamma}} \coloneqq V(\vec{\beta}, \vec{\gamma}) \ket{\iota} \coloneqq \left(\prod_{q} \um(\beta_{q}) \up(\gamma_{q})\right) \ket{\iota}.
    \end{align}
\end{theorem}

\begin{proof}
Let $U_{T}$, $T > 0$, denote the quasi-adiabatic evolution w.r.t.\ $H_{\lin(B_{I}, C)}$.
By \autoref{proposition:EquivalenceOfMixingProperties}, $B_{I}$ is a mixer Hamiltonian in the sense of \autoref{definition:Mixer}.
Therefore, for any $\varepsilon > 0$, \autoref{theorem:ConvergenceOfQAA} implies the existence of a $T > 0$ so that
\begin{align*}
    \norm{(\one - P_{1}(1)) U_{T}(1) \ket{\iota}} < \frac{\varepsilon}{2},
\end{align*}
where $P_{1}$ is the $C^{2}$-continuation of the curve of spectral projections onto the highest energy eigenspaces of $H_{\lin(B_{I}, C)}\vert_{\solspace}$.
W.l.o.g.\ assume that $\dim(\solspace) > 1$ as the statement would be trivial otherwise.
Then, $\alpha \coloneqq \norm{\one - P_{1}(1)}_{\lo(\solspace)} > 0$ since $P_{1}(1)$ has rank one by continuity.
Discretizing the quasi-adiabatic time evolution $U_{T}(1)$ yields the existence of an $m \in \N$ such that
\begin{align}\label{equation:TimeEvolutionEstimate}
    \left\lVert\prod_{j = 1}^{m} e^{- i H_{\lin(B_{I}, C)}\left(j \frac{T}{m}\right) j \frac{T}{m}} - U_{T}(1)\right\rVert < \frac{\varepsilon}{4 \alpha}.
\end{align}
In the following, set 
\begin{align*}
    W_{j} \coloneqq e^{- i H_{\lin(B_{I}, C)}\left(j \frac{T}{m}\right) j \frac{T}{m}}
\end{align*}
and distinguish between the two possibilities to choose a mixer.

\paragraph*{Simultaneous mixer:}
The Lie product formula implies that for all $j \in [m]$, there exist $n_{j} \in \N$ such that for all $\tilde{n} \geq n_{j}$ it holds that 
\begin{align}\label{equation:LieEstimate}
    \left\lVert\left(e^{- i \frac{1 - j \frac{T}{m}}{\tilde{n}} j \frac{T}{m} B_{I}} e^{- i \frac{\left(j \frac{T}{m}\right)^{2}}{\tilde{n}}C}\right)^{\tilde{n}} - W_{j}\right\rVert < \sqrt[m]{\frac{\varepsilon}{4 \alpha} + 1} - 1,
\end{align}
respectively.
Taking $n \coloneqq \max\{n_{j} \defcolon j \in [m]\}$, this estimate holds for all $j \in [m]$ and (especially) $\tilde{n} = n$.

\paragraph*{Sequential mixer:}
W.l.o.g.\ choose the permutation $\sigma = \id_{I}$.
The multivariate Lie product formula \cite[Problem IX.8.5]{Bhatia1997MatrixAnalysis} imples that for all $j \in [m]$, there exist $n_{j} \in \N$ so that for all $\tilde{n} \geq n_{j}$ it holds that 
\begin{align}\label{equation:LieEstimate2}
    \norm{\left(\left(\prod_{i \in I} e^{- i \frac{\left(1 - j \frac{T}{m}\right)}{\tilde{n}} j \frac{T}{m} B_{i}}\right) e^{- i \frac{\left(j \frac{T}{m}\right)^{2}}{\tilde{n}} C}\right)^{\tilde{n}} - W_{j}} < \sqrt[m]{\frac{\varepsilon}{4 \alpha} + 1} - 1,
\end{align}
respectively.

In both cases, choose $q = n^{m}$ parameter values $\vec{\beta} = (\vec{\beta}_{1}, \ldots, \vec{\beta}_{m})$ and $\vec{\gamma} = (\vec{\gamma}_{1}, \ldots, \vec{\gamma}_{m})$ as 
\begin{align*}
    \big(\vec{\beta}_{j}\big)_{k} &= \frac{1 - j \frac{t T}{m}}{n} j \frac{t T}{m}\\
    \big(\vec{\gamma}_{j}\big)_{k} &= \frac{\left(j \frac{t T}{m}\right)^{2}}{n}
\end{align*}
for all $k \in [n]$ and all $j \in [m]$. Then, by construction, \eqref{equation:LieEstimate} and \eqref{equation:LieEstimate2} translate into 
\begin{align*}
    \left\lVert V(\vec{\beta}_{j}, \vec{\gamma}_{j}) - W_{j}\right\rVert < \sqrt[m]{\frac{\varepsilon}{4 \alpha} + 1} - 1 .
\end{align*}
Thus, by \eqref{equation:TimeEvolutionEstimate} and \autoref{lemma:ProductApproximation}, it follows that 
\begin{align*}
    \left\lVert V(\vec{\beta}, \vec{\gamma}) - U_{T}(t)\right\rVert &= \left\lVert\prod_{j = 1}^{m} V(\vec{\beta}_{j}, \vec{\gamma}_{j}) - U_{T}(t)\right\rVert \\
    &\leq \left\lVert\prod_{j = 1}^{m} V(\vec{\beta}_{j}, \vec{\gamma}_{j}) - \prod_{j = 1}^{m} W_{j}\right\rVert  + \left\lVert\prod_{j = 1}^{m} W_{j} - U_{T}(t)\right\rVert \\
    &< \frac{\varepsilon}{4 \alpha} + \frac{\varepsilon}{4 \alpha} = \frac{\varepsilon}{2 \alpha}.
\end{align*}

In summary, it follows that
\begin{align*}
    \norm{(\one - P_{1}(1)) \ket{\vec{\beta}, \vec{\gamma}}} &= \norm{(\one - P_{1}(1)) V(\vec{\beta}, \vec{\gamma}) \ket{\iota}} \\
    &\leq \norm{(\one - P_{1}(1)) (V(\vec{\beta}, \vec{\gamma}) - U_{T}(1)) \ket{\iota}} + \norm{(\one - P_{1}(1)) U_{T}(1) \ket{\iota}} \\
    &< \norm{\one - P_{1}(1)}_{\lo(\solspace)} \frac{\varepsilon}{2 \alpha} + \frac{\varepsilon}{2} = \varepsilon.
\end{align*}
Then, $\im(P_{1}(1)) \subseteq \optspace$ proves the assertion.
\end{proof}

%% file: Tex/Section5.tex
In this paper we presented an elementary proof for the convergence of the QAOA.
This proof can be regarded as a discretized and carefully extended version of the \hyperref[theorem:AdiabaticTheorem]{Adiabatic Theorem}, building on the ideas of Farhi et al.
Beside another core theorem (\hyperref[theorem:PerronFrobenius]{Perron-Frobenius}), this extension is merely based on elementary matrix inequalities.
Most importantly, our proof builds on fewer assumptions (multiple optimal solutions are allowed) and extends to non-trivial feasibility structures ($\solspace \subsetneq \hil$).

Furthermore, the proof canonically gave rise to refined definitions of the \hyperref[definition:Mixer2]{QAOA-mixer} and \hyperref[definition:PhaseSeparator]{QAOA-phase separator} concepts.
Most notably, exactly the same notions arise when properly recreating classical feasibility symmetries within the framework of QAOA (see \cite{Kossmann2022OpenShopSchedulingWithHardConstraints}).
This strongly indicates that the definitions we gave in this paper optimally capture the overall principle the QAOA is based on.

We essentially showed that irreducibility and component-wise non-negativity of the mixer Hamiltonian $B$, restricted to the feasible subspace $\solspace$, are sufficient criteria for the convergence of the QAA and the QAOA.
Moreover, one can readily verify that irreducibility is also a necessary condition in the following sense: Given an arbitrary initial state $\ket{\iota}$ and the existence of a non-trivial $B\vert_{\solspace}$-invariant coordinate subspace, there always exists an objective Hamiltonians $C$ such that the QAA and the QAOA will not be able to approximate any state in $\solspace_{\max}$ to arbitrary precision.
On the other hand, the condition that $B\vert_{\solspace}$ should be component-wise non-negative is not necessary.
In our convergence proof, we imposed this condition in order to apply the \hyperref[theorem:PerronFrobenius]{Perron-Frobenius Theorem}.
However, there also exist more general versions of this theorem (see, e.g., \cite{Schneider1970CrossPositiveMatrices}) which substitute this condition and irreducibility with the more general properties of preserving a given cone and permuting its faces, respectively.
Unfortunately, the cones in questions are merely given by all the orthants in $\solspace$ since every coordinate subspace of $\solspace$ should be represented on their faces.
This yields again, up to some additionally allowed matrix signatures, the same conditions.
Therefore, we do not see much possibilities for relaxing the assumptions, made in \autoref{theorem:ConvergenceOfQAA}  and \autoref{theorem:ConvergenceOfQAOA}.

An interesting and still remaining question is whether one can also characterize the rate of convergence of the QAA for the case of multiple optimal solutions.
In this case, the spectral gap is necessarily vanishing for $t \to 1$ (see \autoref{figure:LevelCrossing}), but stays finite throughout the interval $[0, 1)$.
That is, even though a level crossing occurs, it only happens once and at a predictable time.
There are some results on the rate of convergence in the \hyperref[theorem:AdiabaticTheorem]{Adiabatic Theorem} which are valid for all kinds of (allowed) level crossings (see, e.g., \cite{Avron1999AdiabaticTheoremWithoutAGapCondition,Teufel2001ANoteOnTheAdiabaticTheoremWithoutGapCondition}).
Fine-tuning these results with respect to the particular situation of the QAA promises to be an insightful future project.